\documentclass[reqno,11pt]{amsart}

\usepackage{amsmath}
\usepackage{latexsym}
\usepackage{amssymb}
\usepackage{hyperref}
\usepackage{graphicx}
\usepackage{epstopdf}
\usepackage{ifpdf}
\ifpdf  \fi

\makeatletter
 
 \newcommand{\Rmnum}[1]{\expandafter\@slowromancap\romannumeral #1@}
 \makeatother

\newtheorem{theorem}{Theorem}[section]

\newtheorem{proposition}[theorem]{Proposition}
\newtheorem{lemma}[theorem]{Lemma}

\newcommand{\R}{{\mathbb R}}

\newcommand{\C}{{\mathbb C}}

\newcommand{\be}{\begin{equation}}
\newcommand{\ee}{\end{equation}}
\newcommand{\bea}{\begin{eqnarray}}
\newcommand{\eea}{\end{eqnarray}}
\newcommand{\ba}{\begin{array}}
\newcommand{\ea}{\end{array}}

\newcommand{\ol}{\overline}

\newcommand{\ra}{\rightarrow}

\newcommand{\id}{\mathbb{I}}

\newcommand{\re}{\mathrm{Re}}
\newcommand{\im}{\mathrm{Im}}

\newcommand{\eps}{\varepsilon}

\newcommand{\lam}{\lambda}
\newcommand{\Lam}{\Lambda}

\newcommand{\Gam}{\Gamma}

\newcommand{\Dta}{\Delta}

\numberwithin{equation}{section}

\begin{document}
\title[Large n-limit for Random matrices with External Source with 3 eigenvalues]{Large n-limit for Random matrices with External Source with three distinct eigenvalues}

\author[J.Xu]{Jian Xu}
\address{College of Science,\\
University of Shanghai for Science and Technology,\\
Shanghai 200093\\
People's  Republic of China}
\email{jianxu@usst.edu.cn or jianxu02@gmail.com}

\author[E.Fan]{Engui Fan*}
\address{School of Mathematical Sciences, Institute of Mathematics and Key Laboratory of Mathematics for Nonlinear Science\\
Fudan University\\
Shanghai 200433\\
People's  Republic of China}
\email{correspondence author:faneg@fudan.edu.cn}

\author[Y. Chen]{Yang Chen}
\address{Department of Mathematics,
University of Macau, Macau, P. R. China}
\email{yayangchen@umac.mo}
\keywords{Riemann-Hilbert problem, Random matrices, Multiple
Orthogonal Polynomials, Large N limit}

\date{\today}

\begin{abstract}
In this paper, we analyze the large n-limit for random matrix with
external source with three distinct eigenvalues. And we
confine ourselves in the Hermite case and the three distinct
eigenvalues are $-a,0,a$. For the case $a^2>3$, we establish the universal behavior
of local eigenvalue correlations in the limit $n\rightarrow \infty$, which is known from unitarily
invariant random matrix models. Thus, local eigenvalue correlations are expressed in
terms of the sine kernel in the bulk and in terms of the Airy kernel at the edge of the
spectrum. The result can be obtained by analyzing $4\times 4$
Riemann-Hilbert problem via nonlinear steepest decent method.
\end{abstract}


\maketitle


\section{Introduction and Statement of Results}
We will consider the ensemble of $n\times n$ Hermitian matrices  $M$ with
the density function defined by
\begin{equation}\label{ensem}
\frac{1}{Z_n}e^{-n Tr(V(M)-AM)}dM,
\end{equation}
where $V(M)$ is a  matrix polynomial  with respect to  $M$,  The number $n$ is a large parameter in the ensemble, $A$ is a $n\times n$ matrix, and  $Z_n$  is the normalization constant,
\begin{equation}
Z_n=\int{e^{-Tr(V(M)-AM)}dM},
\end{equation}
where the integral is over all $n\times n$ Hermitian matrices. When
$A=0$, the above ensemble is the standard unitary-invariant ensemble
in the theory of random matrices \cite{m}. The case when $A\ne 0$ is
called a random matrix model with external source, and has been
studied in, for example,
\cite{bkcmp1,bkcmp2,bkcmp3,bh,bh1,bh2,bh3,pzj1,pzj2}. In a series of paper \cite{bkcmp1,bkcmp2,bkcmp3}, Bleher and Kuijlaars considered the large n limit of this model when external source $A$ has two distinct eigenvalues via nonlinear steepest descent method based on a $3\times 3$ Riemann-Hilbert problem.
\par
In this paper, we consider the case external source $A$ has three distinct eigenvalues, say $\{a_1,a_2,a_3\}$. It was shown that when $A$ has three distinct eigenvalues, the correlation kernel function $K_n(x,y)$ is linked to a $4\times 4$ Riemann-Hilbert problem: finding $Y:\C\backslash \R \rightarrow \C^{4\times 4}$ such that
\begin{enumerate}

\item $Y$ is analytic on $\C\backslash \R$,

\item for $x\in \R$,
\begin{equation}\label{ran-Yjump}
Y_+(x)=Y_-(x)\left(
\begin{array}{cccc}1&\omega_1(x)&\omega_2(x)&\omega_3(x)\\0&1&0&0\\0&0&1&0\\0&0&0&1\end{array}\right),
\end{equation}
where $Y_+(x)$ and $Y_-(x)$  denote  the limit of $Y(z)$ as $z\rightarrow x$ from upper and lower half plane,respectively.

\item as $z\rightarrow \infty$, we have
\begin{equation}\label{ran-Yasy}
Y(z)=\left(\id+O(\frac{1}{z})\right)diag{(z^n,z^{-n_1},z^{-n_2},z^{-n_3})},
\end{equation}
where $\id$ denotes the $4\times 4$ identity matrix.

\end{enumerate}

We consider correlation kernel
\be\label{ran-KnY}
K_n(x,y)=\frac{e^{-\frac{1}{2}(V(x)+V(y))}}{2\pi i (x-y)}(\ba{cccc}0&e^{na_1y}&e^{na_2y}&e^{na_3y}\ea)Y^{-1}(y)Y(x)\left(\ba{c}1\\0\\0\\0\ea\right),
\ee
and  assume that the external source $A$ is a fixed diagonal matrix with $n_1$ eigenvalues $a_1=a$, $n_2$ eigenvalues $a_2=0$ and $n_3$ eigenvalues $a_3=-a$. It is the aim of this paper to analyze the Riemann-Hilbert problem as $n\rightarrow \infty$, by using the method of nonlinear steepest descent of Deift and Zhou \cite{dz}. And we focus here on Gaussian case $V(x)=\frac{1}{2}x^2$.
The main results of this paper are
\begin{theorem}
Let $V(M)=\frac{1}{2}M^2$, $n_1=n_3$, $n_2=t$, $a^2>3$. Then
the limiting mean density of eigenvalues
\be
\rho(z)=\lim_{n\ra\infty}\frac{1}{n}K_n(x,x),
\ee
exists, and it is supported by three intervals, $[-z_3,-z_2], [-z_1,z_1]$ and $[z_2,z_3]$.
\end{theorem}
\begin{theorem}\label{thm2}
For every $x_0\in(-z_3,-z_2)\cup (-z_1,z_1)\cup (z_2,z_3)$ and $u,v\in \R$, we have
\be
\lim_{n\ra\infty}\frac{1}{n\rho(x_0)}\hat K_n\left(x_0+\frac{u}{n\rho(x_0)},x_0+\frac{v}{n\rho(x_0)}\right)=\frac{sin \pi(u-v)}{\pi(u-v)}.
\ee
\end{theorem}
\begin{theorem}\label{ran-thm3}
We use the same notation as in Theorem \ref{thm2}.
For every $u,v\in \R$, we have
\be
\lim_{n\ra\infty}\frac{1}{(\rho_1n)^{\frac{2}{3}}}\hat K_n\left(z_3+\frac{u}{(\rho_1n)^{\frac{2}{3}}},z_3+\frac{v}{(\rho_1n)^{\frac{2}{3}}}\right)=\frac{Ai(u)Ai'(v)-Ai'(u)Ai(v)}{u-v},
\ee
\be
\lim_{n\ra\infty}\frac{1}{(\rho_2n)^{\frac{2}{3}}}\hat K_n\left(z_2+\frac{u}{(\rho_2n)^{\frac{2}{3}}},z_2+\frac{v}{(\rho_2n)^{\frac{2}{3}}}\right)=\frac{Ai(u)Ai'(v)-Ai'(u)Ai(v)}{u-v},
\ee
\be
\lim_{n\ra\infty}\frac{1}{(\rho_3n)^{\frac{2}{3}}}\hat K_n\left(z_1+\frac{u}{(\rho_3n)^{\frac{2}{3}}},z_1+\frac{v}{(\rho_3n)^{\frac{2}{3}}}\right)=\frac{Ai(u)Ai'(v)-Ai'(u)Ai(v)}{u-v},
\ee
Similar limits hold near the edge points $-z_3,-z_2,-z_1$.
\end{theorem}

\section{Differential Equations}
In \cite{agk}, it was shown that the Riemann-Hilbert problem of $Y$ has a unique solution
\begin{equation}
Y=\left(
\begin{array}{cccc}
P_{n_1,n_2,n_3}&C(P_{n_1,n_2,n_3}\omega_1)&C(P_{n_1,n_2,n_3}\omega_2)&C(P_{n_1,n_2,n_3}\omega_3)\\
c_1P_{n_1-1,n_2,n_3}&c_1C(P_{n_1-1,n_2,n_3}\omega_1)&c_1C(P_{n_1-1,n_2,n_3}\omega_2)&c_1C(P_{n_1-1,n_2,n_3}\omega_3)\\
c_2P_{n_1,n_2-1,n_3}&c_2C(P_{n_1,n_2-1,n_3}\omega_1)&c_2C(P_{n_1,n_2-1,n_3}\omega_2)&c_2C(P_{n_1,n_2-1,n_3}\omega_3)\\
c_3P_{n_1,n_2,n_3-1}&c_3C(P_{n_1,n_2,n_3-1}\omega_1)&c_3C(P_{n_1,n_2,n_3-1}\omega_2)&c_3C(P_{n_1,n_2,n_3-1}\omega_3)
\end{array}
\right)
\end{equation}
where $P_{n_1,n_2,n_3}(x)=x^n+\cdots$ is a monic orthogonal  polynomial of degree $n=n_1+n_2+n_3$  with constants
\begin{equation}\label{cicoeff}
c_1=-2\pi i(h^{(1)}_{n_1-1,n_2,n_3})^{-1},\quad c_2=-2\pi
i(h^{(2)}_{n_1,n_2-1,n_3})^{-1},\quad c_3=-2\pi
i(h^{(3)}_{n_1,n_2,n_3-1})^{-1}
\end{equation}
and where $h^{(j)}_{k_1,k_2,k_3}=\int_{-\infty}^{\infty}P_{k_1,k_2,k_3}x^{k_j}\omega_j(x)dx,j=1,2,3$ and $\omega_j(x)=e^{-(V(x)-a_j x)},\quad j=1,2,3,$ and $Cf$ denotes the Cauchy transform of $f$, that is,
\begin{equation}
Cf(z)=\frac{1}{2\pi i}\int_{\R}\frac{f(s)}{s-z}ds.
\end{equation}
Introducing
\begin{equation}\label{Psidef}
\Psi_{n_1,n_2,n_3}=\left(
\begin{array}{cccc}
1&0&0&0\\
0&\frac{1}{c_1}&0&0\\
0&0&\frac{1}{c_2}&0\\
0&0&0&\frac{1}{c_3}
\end{array}
\right)Y_{n_1,n_2,n_3}\left(\begin{array}{cccc}e^{-V(x)}&0&0&0\\0&e^{-a_1x}&0&0\\0&0&e^{-a_2x}&0\\0&0&0&e^{-a_3x}\end{array}\right).
\end{equation}
we can obtain
\begin{enumerate}
\item $\Psi_{n_1,n_2,n_3}$ is analytic on $\C\backslash\R$,

\item for $x\in \R$, we have
\begin{equation}\label{Psijump}
\Psi_{n_1,n_2,n_3,+}(x)=\Psi_{n_1,n_2,n_3,-}(x)\left(\begin{array}{cccc}1&1&1&1\\0&1&0&0\\0&0&1&0\\0&0&0&1\end{array}\right),
\end{equation}

\item as $z\rightarrow \infty$, we also have
\begin{equation}\label{Psiasy}
\Psi(z)=\left(\id+\frac{\Psi^{(1)}_{n_1,n_2,n_3}}{z}+O(\frac{1}{z^2})\right)\left(\begin{array}{cccc}z^ne^{-V(x)}&0&0&0\\0&\frac{1}{c_1}z^{-n_1}e^{-a_1x}&0&0\\0&0&\frac{1}{c_2}z^{-n_2}e^{-a_2x}&0\\0&0&0&\frac{1}{c_3}z^{-n_3}e^{-a_3x}\end{array}\right)
\end{equation}
where
\begin{equation}\label{Psi1}
\Psi^{(1)}_{n_1,n_2,n_3}=\left(\begin{array}{cccc}p_{n_1,n_2,n_3}&\frac{h^{(1)}_{n_1,n_2,n_3}}{h^{(1)}_{n_1-1,n_2,n_3}}&\frac{h^{(2)}_{n_1,n_2,n_3}}{h^{(2)}_{n_1,n_2-1,n_3}}&\frac{h^{(3)}_{n_1,n_2,n_3}}{h^{(3)}_{n_1,n_2,n_3-1}}\\
1&\frac{q^{(1)}_{n_1-1,n_2,n_3}}{h^{(1)}_{n_1-1,n_2,n_3}}&\frac{h^{(2)}_{n_1-1,n_2,n_3}}{h^{(2)}_{n_1,n_2-1,n_3}}&\frac{h^{(3)}_{n_1-1,n_2,n_3}}{h^{(3)}_{n_1,n_2,n_3-1}}\\
1&\frac{h^{(1)}_{n_1,n_2-1,n_3}}{h^{(1)}_{n_1-1,n_2,n_3}}&\frac{q^{(2)}_{n_1,n_2-1,n_3}}{h^{(2)}_{n_1,n_2-1,n_3}}&\frac{h^{(3)}_{n_1,n_2-1,n_3}}{h^{(3)}_{n_1,n_2,n_3-1}}\\
1&\frac{h^{(1)}_{n_1,n_2,n_3-1}}{h^{(1)}_{n_1-1,n_2,n_3}}&\frac{h^{(2)}_{n_1,n_2,n_3-1}}{h^{(2)}_{n_1,n_2-1,n_3}}&\frac{q^{(3)}_{n_1,n_2,n_3-1}}{h^{(3)}_{n_1,n_2,n_3-1}}
\end{array}\right)
\end{equation}
here
\begin{equation}\label{qkdef}
q^{(k)}_{n_1,n_2,n_3}=\int_{-\infty}^{\infty}P_{n_1,n_2,n_3}(x)x^{n_k+1}\omega_k(x)dx,\quad k=1,2,3.
\end{equation}
and $P_{n_1,n_2,n_3}(z)=z^n+p_{n_1,n_2,n_3}z^{n-1}+\cdots$ is a multiple polynomials.
\end{enumerate}

According to the above notations, we have the following property.
\begin{proposition}
We have the differential equation,

\be\label{diffeqn}
 \Psi'_{n_1,n_2,n_3}(z)=N\left(
                          \ba{cccc}
                           -z&\frac{n_1}{N}&\frac{n_2}{N}&\frac{n_3}{N}\\
                           -1&-a&0&0\\
                           -1&0&0&0\\
                           -1&0&0&a
                          \ea
                          \right)\Psi_{n_1,n_2,n_3}(z).
\ee
where the prime $'$ denotes the derivative with respect to $z$.
\end{proposition}
\begin{proof}
Let
\be
A_{n_1,n_2,n_3}=\frac{1}{N}\Psi'_{n_1,n_2,n_3}\Psi_{n_1,n_2,n_3}
\ee
From (\ref{Psijump}), we know $A_{n_1,n_2,n_3}$ has no jump, that is to say $A_{n_1,n_2,n_3}$ is analytic on the whole complex plane. From (\ref{Psiasy}), we can show as $z\ra \infty$
\be
A_{n_1,n_2,n_3}=\left(\id+\frac{\Psi^{(1)}_{n_1,n_2,n_3}}+\dots\right)\left(\ba{cccc}-z&0&0&0\\0&-a_1&0&0\\0&0&-a_2&0\\0&0&0&-a_4\ea\right)\left(\id+\frac{\Psi^{(1)}_{n_1,n_2,n_3}}+\dots\right)^{-1}+O(\frac{1}{z}).
\ee
As $A_{n_1,n_2,n_3}$ is entire function, it implies that
\be
A_{n_1,n_2,n_3}=\left(\ba{cccc}-z&c_{n_1,n_2,n_3}&d_{n_1,n_2,n_3}&e_{n_1,n_2,n_3}\\-1&-a_1&0&0\\-1&0&-a_2&0\\-1&0&0&-a_3\ea\right),
\ee
where $c_{n_1,n_2,n_3},d_{n_1,n_2,n_3},e_{n_1,n_2,n_3}$ are needed to determine.

\par
From the differential equation we get
\be
\Psi'_{n_1,n_2,n_3}(z)=NA_{n_1,n_2,n_3}\Psi_{n_1,n_2,n_3}(z).
\ee
Let us show how to determine $c_{n_1,n_2,n_3},d_{n_1,n_2,n_3},e_{n_1,n_2,n_3}$. From the recursion formula of $\Psi_{n_1,n_2,n_3}(z)$ \cite{jx},
say $\Psi_{n_1+1,n_2,n_3}=U_{n_1,n_2,n_3}\Psi_{n_1,n_2,n_3}$, where
\begin{equation}\label{Un1def}
\begin{array}{rl}
U_{n_1,n_2,n_3}
&=\left(\begin{array}{cccc}z-b_{n_1,n_2,n_3}&-c_{n_1,n_2,n_3}&-d_{n_1,n_2,n_3}&-e_{n_1,n_2,n_3}\\
1&0&0&0\\1&0&f_{n_1,n_2,n_3}&0\\1&0&0&g_{n_1,n_2,n_3}\end{array}\right)
\end{array}
\end{equation}
we can get
$$\Psi'_{n_1+1,n_2,n_3}=U'_{n_1,n_2,n_3}\Psi_{n_1,n_2,n_3}+U_{n_1,n_2,n_3}\Psi'_{n_1,n_2,n_3},$$
on the other side, we have
$$\Psi'_{n_1+1,n_2,n_3}(z)=NA_{n_1+1,n_2,n_3}\Psi_{n_1+1,n_2,n_3}(z).$$
This implies
$$U'_{n_1,n_2,n_3}+U_{n_1,n_2,n_3}NA_{n_1,n_2,n_3}=NA_{n_1+1,n_2,n_3}U_{n_1,n_2,n_3}.$$
A simpler calculation shows that
\be
\ba{ll}
b_{n_1,n_2,n_3}=a_1,& c_{n_1+1,n_2,n_3}=c_{n_1,n_2,n_3}+\frac{1}{N},\\
 d_{n_1+1,n_2,n_3}=d_{n_1,n_2,n_3},& e_{n_1+1,n_2,n_3}=e_{n_1,n_2,n_3},\\
 f_{n_1,n_2,n_3}=a_2-a_1,&g_{n_1,n_2,n_3}=a_3-a_1.
\ea
\ee
Similarly, we also can show the following formulas via analyze the other recursion relations of $\Psi_{n_1,n_2,n_3}(z)$,
\be
\ba{ll}
\tilde b_{n_1,n_2,n_3}=a_2,&\tilde c_{n_1+1,n_2,n_3}=\tilde c_{n_1,n_2,n_3},\\
 \tilde d_{n_1+1,n_2,n_3}=\tilde d_{n_1,n_2,n_3}+\frac{1}{N},&\tilde e_{n_1+1,n_2,n_3}=\tilde e_{n_1,n_2,n_3},\\
 \tilde f_{n_1,n_2,n_3}=a_1-a_2,&\tilde g_{n_1,n_2,n_3}=a_3-a_2.
\ea
\ee
\be
\ba{ll}
\hat b_{n_1,n_2,n_3}=a_3,&\hat c_{n_1+1,n_2,n_3}=\hat c_{n_1,n_2,n_3},\\
 \hat d_{n_1+1,n_2,n_3}=\hat d_{n_1,n_2,n_3},&\hat e_{n_1+1,n_2,n_3}=\hat e_{n_1,n_2,n_3}+\frac{1}{N},\\
 \hat f_{n_1,n_2,n_3}=a_1-a_3,&\hat g_{n_1,n_2,n_3}=a_2-a_2.
\ea
\ee

By the initial value $c_{0,n_2,n_3}=d_{n_1,0,n_3}=e_{n_1,n_2,0}=0$, we can get
\be
c_{n_1,n_2,n_3}=\frac{n_1}{N},\quad d_{n_1,n_2,n_3}=\frac{n_2}{N},\quad e_{n_1,n_2,n_3}=\frac{n_3}{N}.
\ee
Noticing our assumption $a_1=a,a_2=0,a_3=-a$, we end the proof.

\end{proof}

We look for a WKB solution of the differential equation
(\ref{diffeqn}) of the form

\be\label{WKBsol}
 \Psi_{n_1,n_2,n_3}(z)=W(z)e^{-N\Lam(z)},
\ee where $\Lam$ is a diagonal matrix. By substituting this form
into (\ref{diffeqn}) we obtain the equation,

\begin{equation}\label{WKBeqn}
-W\Lam'W^{-1}=A+\frac{1}{N}W'W^{-1},
\end{equation}
where $A$ is the matrix of coefficients in (\ref{diffeqn}). By
dropping the last term we reduce it to the eigenvalue problem,

\begin{equation}\label{WKBeigen}
W\Lam'W^{-1}=-A.
\end{equation}
The characteristic polynomial is

\begin{equation}\label{charact}
\ba{rl}
 det[\xi \id+A]=&{}\left|
                  \ba{cccc}
                   \xi-z&t_1&t_2&t_3\\
                   -1&\xi-a_1&0&0\\
                   -1&0&\xi-a_2&0\\
                   -1&0&0&\xi-a_3
                   \ea
                   \right|\\
               =&{}\xi^4-(z+a_1+a_2+a_3)\xi^3\\
               &{}+[(t_1+t_2+t_3)+a_1a_2+a_1a_3+a_2a_3+z(a_1+a_2+a_3)]\xi^2\\
               &{}-[(a_1a_2+a_1a_3+a_2a_3)z+(t_1+t_2)a_3+(t_1+t_3)a_2+(t_2+t_3)a_1+a_1a_2a_3]\xi\\
               &{}+a_1a_2t_3+a_1t_2a_3+t_1a_2a_3+a_1a_2a_3z=0
\ea
\end{equation}
where $t_1=\frac{n_1}{N},t_2=\frac{n_2}{N}$ and $t_3=\frac{n_3}{N}$.

\par
The above equation defines a Riemann surface, which in the case of
interest in this paper (where $n_1=n_3, a_1=-a_3=a, a_2=0$) reduces
to

\begin{equation}\label{Specurve}
\xi^4-z\xi^3+(1-a^2)\xi^2+a^2z\xi-t_2a^2=0.
\end{equation}
This defines the Riemann surface that will play a central role in
the rest of the paper.

\section{Spectral Curve and Riemann Surface}
The Riemann surface is given by (\ref{Specurve}) or, if we solve for
$z$,

\begin{equation}\label{Specurvez}
z=\frac{\xi^4+(1-a^2)\xi^2-t_2a^2}{\xi^3-a^2\xi}.
\end{equation}
There are four inverse functions to (\ref{Specurvez}), which we
choose such that as $z\rightarrow \infty$,

\begin{equation}\label{fourxi}
\ba{l}
 \xi_1(z)=z-\frac{1}{z}+O(\frac{1}{z^3}),\\
 \xi_2(z)=a+\frac{t_1}{z}+O(\frac{1}{z^2}),\\
 \xi_3(z)=\frac{t_2}{z}+O(\frac{1}{z^2}),\\
 \xi_4(z)=-a+\frac{t_3}{z}+O(\frac{1}{z^2}).
\ea
\end{equation}

\par
We need to find the sheet structure of the Riemann surface
(\ref{Specurve}). The critical points of $z(\xi)$ satisfy the
equation

\begin{equation}\label{crieqn}
\xi^6-(1+2a^2)\xi^4+[a^4+(3t_2-1)a^2]\xi^2-t_2a^4=0,
\end{equation}
Let $y=\xi^2$, then the above equation becomes

\begin{equation}\label{crieqnasy}
y^3-(1+2a^2)y^2+[a^4+(3t_2-1)a^2]y-t_2a^4=0
\end{equation}
which is a cubic equation as $y$. The discriminant of this cubic
equation is, and we denote $a^2$ as $b$ and $t_2$ as $t$,

\begin{equation}\label{disminant}
\Dta=(1-t)b^2\Dta_c=(1-t)b^2[8b^3-(9t+15)b^2+(108t^2-90t+6)b+(1-9t)]
\end{equation}
and then we write $\Dta_c$ as a function of $t$,

\begin{equation}\label{disminantc}
\Dta_c(t)=108bt^2+(-9b^2-90b-9)t+8b^3-15b^2+6b+1
\end{equation}
This is a quadratic function of $t$. And the discriminant of this
quadratic function is

\begin{equation}
\Dta_q=-3(b-3)(5b+1)^3,
\end{equation}

\par
We can obtain that if $b>3$, then $\Dta_c>0$ for all $t\in (0,1)$,
then we have $\Dta>0$. That means the cubic equation
(\ref{crieqnasy})of $y$ has three distinct real roots. And in here,
we can go further beyond this, that is, the three distinct real
roots are all positive. We denote the roots by $y_1,y_2$ and $y_3$
(without loss of generally, we can set $y_1<y_2<y_3$), and then set

\begin{equation}\label{roots}
p=\sqrt{y_1},\quad q=\sqrt{y_2},\quad r=\sqrt{y_3}.
\end{equation}

Then the critical points on the $\xi-$plane are $\pm p,\pm q$ and
$\pm r$. The branch points on the $z-$plane are $\pm z_1,\pm z_2$
and $\pm z_3$, where

\begin{equation}\label{zjdef}
z_1=z(\xi=p),\quad z_2=z(\xi=q),\quad z_3=z(\xi=r),\quad
0<z_1<z_2<z_3.
\end{equation}
\par
We can show that $\xi_1,\xi_2,\xi_3$ and $\xi_4$ have analytic extensions
to $\C\backslash ([-z_3,-z_2]\cup [-z_1,z_1]\cup [z_2,z_3])$,$\C\backslash [z_2,z_3]$,$\C\backslash [-z_1,z_1]$ and $\C\backslash [-z_3,-z_2]$.
On the cut $[z_2,z_3]$,
\be
\ba{l}
\xi_{1+}(x)=\ol{\xi_{1-}(x)}=\xi_{2-}(x)=\ol{\xi_{2+}(x)},\quad z_2\le x\le z_3,\\
\im{\xi_{1+}}>0,\qquad z_2<x<z_3,
\ea
\ee
and $\xi_3(x),\xi_4(x)$ are real. On the cut $[-z_1,z_1]$,
\be
\ba{l}
\xi_{1+}(x)=\ol{\xi_{1-}(x)}=\xi_{3-}(x)=\ol{\xi_{3+}(x)},\quad -z_1\le x\le z_1,\\
\im{\xi_{1+}}>0,\qquad -z_1<x<z_1,
\ea
\ee
$\xi_2(x),\xi_4(x)$ are real. On the cut $[-z_3,-z_2]$,
\be
\ba{l}
\xi_{1+}(x)=\ol{\xi_{1-}(x)}=\xi_{4-}(x)=\ol{\xi_{4+}(x)},\quad -z_3\le x\le -z_2,\\
\im{\xi_{1+}}>0,\qquad -z_3<x<-z_2,
\ea
\ee
$\xi_2(x),\xi_3(x)$ are real. See the figure \ref{ran-fig1}.
\begin{figure}[th]
\centering
\includegraphics{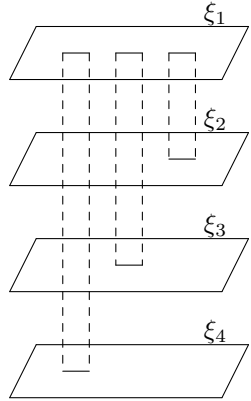}
\caption{Four sheet structure of Riemann surface}\label{ran-fig1}
\end{figure}

\par

We define
\be
\rho(x)=\frac{1}{\pi}\im \xi_{1+}(x),\qquad x\in [-z_3,-z_2]\cup [-z_1,z_1]\cup [z_2,z_3],
\ee
then as $x\in (-z_3,-z_2)\cup (-z_1,z_1)\cup (z_2,z_3)$, we have $\rho(x)>0$ and
\be
\ba{rl}
\int_{z_2}^{z_3}\rho(x)dx&=\int_{z_2}^{z_3}\frac{1}{\pi}\im \xi_{1+}(x)dx\\
{}&=\int_{z_2}^{z_3}\frac{1}{2\pi i}(\xi_{1+}-\ol{\xi_{1+}})(x)dx\\
{}&=\int_{z_2}^{z_3}\frac{1}{2\pi i}(\xi_{2-}-\xi_{2+})(x)dx\\
{}&=\frac{1}{2\pi i}\oint\xi_{2}(s)ds=\frac{1-t}{2}.
\ea
\ee

\par

Similarly, we have
\be
\int_{-z_1}^{z_1}\rho(x)dx=t,\quad \int_{-z_3}^{-z_2}\rho(x)dx=\frac{1-t}{2}.
\ee
And there exists $\rho_j>0,j=1,2,3$ such that
\be
\ba{l}
\rho(x)=\frac{\rho_j}{\pi}|x-z_j|^{\frac{1}{2}}\left(1+O(x-z_j)\right),\quad x\ra z_j,x\in (z_2,z_3),\\
\rho(x)=\frac{\rho_j}{\pi}|x+z_j|^{\frac{1}{2}}\left(1+O(x+z_j)\right),\quad x\ra -z_j,x\in (-z_3,-z_2),\\
\rho(x)=\frac{\rho_j}{\pi}|x-z_j|^{\frac{1}{2}}\left(1+O(x-z_j)\right),\quad x\ra z_j,x\in (-z_1,z_1).
\ea
\ee

The reason is there exists a constant $\rho_3>0$ near the branch point $z_3$ such that as $z\ra z_3$,
\be
\ba{l}
\xi_{1}(z)=r+\rho_3(z-z_3)^{\frac{1}{2}}+O(z-z_3),\\
\xi_2(z)=r-\rho_3(z-z_3)^{\frac{1}{2}}+O(z-z_3).
\ea
\ee
Similarly, there exists a constant $\rho_2>0$ near the branch point $z_2$ such that as $z\ra z_2$,
\be
\ba{l}
\xi_{1}(z)=q-\rho_2(z-z_2)^{\frac{1}{2}}+O(z-z_2),\\
\xi_2(z)=q+\rho_2(z-z_2)^{\frac{1}{2}}+O(z-z_2).
\ea
\ee
There exists a constant $\rho_1>0$ near the branch point $z_1$ such that as $z\ra z_1$,
\be
\ba{l}
\xi_{1}(z)=p-\rho_1(z-z_1)^{\frac{1}{2}}+O(z-z_1),\\
\xi_2(z)=p+\rho_1(z-z_1)^{\frac{1}{2}}+O(z-z_1).
\ea
\ee

\par
Next, we need the integrals of the $\xi-$function,
\be
\lam_j(z)=\int^{z}\xi_j(s)ds,\qquad j=1,2,3,4.
\ee
Here, we choose the integral path so that $\lam_1(z)$ and $\lam_2(z)$ are analytic at $\C\backslash (-\infty,z_3]$, $\lam_3(z)$ is analytic at $\C\backslash (-\infty,z_1]$ and $\lam_4(z)$ is analytic at $\C\backslash (-\infty,-z_2]$. From (\ref{fourxi}), we know as $z\ra \infty$
\be\label{fourlam}
\ba{l}
\lam_1(z)=\frac{1}{2}z^2-\ln{z}+l_1+O(\frac{1}{z^2}),\\
\lam_2(z)=az+\frac{1-t}{2}\ln{z}+l_2+O(\frac{1}{z}),\\
\lam_3(z)=t\ln{z}+l_3+O(\frac{1}{z}),\\
\lam_4(z)=-az+\frac{1-t}{2}\ln{z}+l_4+O(\frac{1}{z}),
\ea
\ee
where $l_1,l_2,l_3,l_4$ are constants, which we choose as follows. We choose $l_1,l_2$ such that
\be
\lam_1(z_3)=\lam_2(z_3)=0,
\ee
choose $\l_3$ and $l_4$ such that
\be
\lam_3(z_1)=\lam_{1+}(z_1)=\lam_{1-}(z_1)-(1-t)\pi i,
\ee
\be
\lam_{4}(-z_2)=\lam_{1+}(-z_2)=\lam_{1-}(-z_2)-(1+t)\pi i.
\ee
Then we have the following jump relations:
\be\label{lamjump}
\ba{ll}
\lam_{1+}(x)=\lam_{2-}(x),\quad \lam_{1-}(x)=\lam_{2+}(x),&x\in [z_2,z_3],\\
\lam_{2+}(x)-\lam_{2-}(x)=(1-t)\pi i,&x\in (-\infty,z_2],\\
\lam_{1+}(x)-\lam_{1-}(x)=-(1-t)\pi i,&x\in [z_1,z_2],\\
\lam_{1+}(x)=\lam_{3-}(x),\quad \lam_{1-}(x)-\lam_{3+}(x)=(1-t)\pi i,&x\in [-z_1,z_1],\\
\lam_{1+}(x)-\lam_{1-}(x)=-(1+t)\pi i,\quad \lam_{3+}(x)-\lam_{3-}(x)=2t\pi i,&x\in [-z_2,-z_1],\\
\lam_{1+}(x)=\lam_{4-}(x),\quad \lam_{1-}(x)-\lam_{4+}(x)=(1+t)\pi i,&x\in [-z_3,-z_2],\\
\lam_{1+}(x)-\lam_{1-}(x)=-2\pi i,\quad \lam_{4+}(x)-\lam_{4-}(x)=(1-t)\pi i,&x\in (-\infty,-z_3].
\ea
\ee

\par
For later use, we state the following propositions.
\begin{lemma}\label{lamdaxiao}
(a)The open interval $(z_2,z_3)$ has a neighborhood $U_1$ in the complex
plane such that when $z\in U_1\backslash (z_2,z_3)$, the real part $\re \lam_2$ is the biggest one among $\lam_j,j=1,2,3,4$;

(b)The open interval $(-z_1,z_1)$ has a neighborhood $U_2$ in the complex
plane such that when $z\in U_1\backslash (-z_1,z_1)$, the real part $\re \lam_3$ is the biggest one among $\lam_j,j=1,2,3,4$;

(c)The open interval $(-z_3,-z_2)$ has a neighborhood $U_3$ in the complex
plane such that when $z\in U_1\backslash (-z_3,-z_2)$, the real part $\re \lam_4$ is the biggest one among $\lam_j,j=1,2,3,4$.
\end{lemma}

\begin{proof}
We just prove (a), the following (b) and (c) are similarly.

\par
First, we have $\re \lam_{3+}<\re \lam_{1-},\re \lam_{4+}<\lam_{1-}$ when $z\in U_1\backslash [z_2,z_3]$. Then, define $F(x)=\lam_{2+}(x)-\lam_{1+}(x)$, we have
$F(x)=\ol{\lam_{1+}(x)}-\lam_{1+}(x)$. That is to say, $F(x)$ is pure image on the interval $(z_2,z_3)$. As we have $F'(x)=\xi_{2+}(x)-\xi_{1+}(x)=\xi_{2+}(x)-\xi_{2-}(x)=-2\pi i\rho(x)$, then $\im F'(x)<0$. By Cauchy-Riemann equation, $\re F(x)>0$ as we move from
the interval $(z_2,z_3)$ into the upper half-plane. $\re F(x)>0$ as we move from
the interval $(z_2,z_3)$ into the lower half-plane.
\end{proof}

\section{Transformations of the Riemann-Hilbert Problem}
\subsection{First Transform}
Using the functions $\lam_j$ and constants $l_j,j=1,2,3,4$, which are defined in the previous section, we define
\be\label{ran-1trans}
T(z)=diag\left(e^{-nl_1},e^{-nl_2},e^{-nl_3},e^{-nl_4}\right)Y(z)diag\left(e^{n(\lam_1(z)-\frac{1}{2}z^2)},e^{n(\lam_2(z)-az)},e^{n\lam_3(z)},e^{n(\lam_4(z)+az)}\right),
\ee
\par
Then by (\ref{ran-Yjump}) and (\ref{ran-1trans}), we have $T_+(x)=T_-(x)J_{T}(x),x\in\R$, where
\be\label{ran-Tjump}
J_{T}(x)=\left(
          \ba{cccc}
          e^{n(\lam_{1+}-\lam_{1-})}&e^{n(\lam_{2+}-\lam_{1-})}&e^{n(\lam_{3+}-\lam_{1-})}&e^{n(\lam_{4+}-\lam_{1-})}\\
          0&e^{n(\lam_{2+}-\lam_{2-})}&0&0\\
          0&0&e^{n(\lam_{3+}-\lam_{3-})}&0\\
          0&0&0&e^{n(\lam_{4+}-\lam_{4-})}
          \ea
          \right)
\ee
Be more exact, on $[z_2,z_3]$
\be
J_{T}(x)=\left(
          \ba{cccc}
          e^{n(\lam_{1+}-\lam_{2+})}&1&e^{n(\lam_{3}-\lam_{1-})}&e^{n(\lam_{4}-\lam_{1-})}\\
          0&e^{n(\lam_{1-}-\lam_{2-})}&0&0\\
          0&0&1&0\\
          0&0&0&1
          \ea
          \right);
\ee
on $[-z_1,z_1]$
\be
J_{T}(x)=\left(
          \ba{cccc}
          e^{n(\lam_{1+}-\lam_{3+})}&e^{n(\lam_{2+}-\lam_{1-})}&1&e^{n(\lam_{4}-\lam_{1-})}\\
          0&1&0&0\\
          0&0&e^{n(\lam_{1-}-\lam_{3-})}&0\\
          0&0&0&1
          \ea
          \right);
\ee
on $[-z_3,-z_2]$
\be
J_{T}(x)=\left(
          \ba{cccc}
          e^{n(\lam_{1+}-\lam_{1-})}&e^{n(\lam_{2+}-\lam_{1-})}&e^{n(\lam_{3+}-\lam_{1-})}&1\\
          0&1&0&0\\
          0&0&1&0\\
          0&0&0&e^{n(\lam_{4+}-\lam_{4-})}
          \ea
          \right);
\ee
on $\R\backslash ([z_2,z_3]\cup [-z_1,z_1]\cup [-z_3,-z_2])$
\be\label{ran-Tjumprealsheng}
J_{T}(x)=\left(
          \ba{cccc}
          1&e^{n(\lam_{2+}-\lam_{1-})}&e^{n(\lam_{3+}-\lam_{1-})}&e^{n(\lam_{4+}-\lam_{1-})}\\
          0&1&0&0\\
          0&0&1&0\\
          0&0&0&1
          \ea
          \right).
\ee
\par

By (\ref{ran-Yasy}), (\ref{fourlam}) and (\ref{ran-1trans}), we know the asymptotics of $T$ is
\be
T(z)=\id+O(\frac{1}{z}),\qquad z\ra \infty.
\ee

\par

Thus $T$ solves the Riemann-Hilbert problem
\begin{itemize}
 \item $T$ is analytic on $\C\backslash \R$,

 \item
  \be\label{ran-TRHP}
  T_+(x)=T_-(x)J_T(x),\quad x\in \R,
  \ee

 \item As $z\ra \infty$,
 \be
 T(z)=\id+O(\frac{1}{z}).
 \ee
\end{itemize}

Inserting (\ref{ran-1trans}) into (\ref{ran-KnY}), we see that the kernel $K_n$ can be expressed in terms of $T$ as
follows:
\be\label{ran-KnT}
K_n(x,y)=\frac{e^{\frac{1}{4}n(x^2-y^2)}}{2\pi i (x-y)}\left(\ba{cccc}0&e^{n\lam_{2+}(y)}&e^{n\lam_{3+}(y)}&e^{n\lam_{4+}(y)}\ea\right)T_+^{-1}(y)T_+(x)\left(\ba{c}e^{-n\lam_{1+}(x)}\\0\\0\\0\ea\right).
\ee

\subsection{Second Transform}
The second transformation of the Rirmann-Hilbert problem is opening of lenses. Consider a lens with
vertices $z_2,z_3$, see figure \ref{ran-fig2}.
\begin{figure}[th]
\centering
\includegraphics{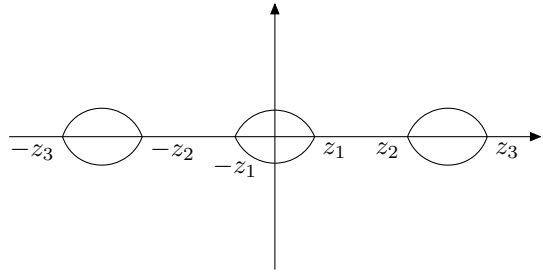}
\caption{The lens of vertices of $\pm z_j,j=1,2,3$.}\label{ran-fig2}
\end{figure}
The lens is contained in the neighborhood $U_1$ of $(z_2,z_3)$. We have the factorization,
\be
\ba{l}
\left(
          \ba{cccc}
          e^{n(\lam_{1+}-\lam_{2+})}&1&e^{n(\lam_{3}-\lam_{1-})}&e^{n(\lam_{4}-\lam_{1-})}\\
          0&e^{n(\lam_{1-}-\lam_{2-})}&0&0\\
          0&0&1&0\\
          0&0&0&1
          \ea
          \right)\\
=\left(\ba{cccc}1&0&0&0\\e^{n(\lam_1-\lam_2)_-}&1&-e^{n(\lam_3-\lam_2)_-}&-e^{n(\lam_4-\lam_2)_-}\\0&0&1&0\\0&0&0&1\ea\right)
\left(\ba{cccc}0&1&0&0\\-1&0&0&0\\0&0&1&0\\0&0&0&1\ea\right)\\
\left(\ba{cccc}1&0&0&0\\e^{n(\lam_1-\lam_2)_+}&1&e^{n(\lam_3-\lam_2)_+}&e^{n(\lam_4-\lam_2)_+}\\0&0&1&0\\0&0&0&1\ea\right).
\ea
\ee
Set
\be\label{ran-2S1}
S(z)=\left\{
\ba{l}
T(z)\left(\ba{cccc}1&0&0&0\\-e^{n(\lam_1-\lam_2)}&1&-e^{n(\lam_3-\lam_2)}&-e^{n(\lam_4-\lam_2)}\\0&0&1&0\\0&0&0&1\ea\right)\mbox{ ͸¾µµÄÉϱßÇøÓò}\\
T(z)\left(\ba{cccc}1&0&0&0\\e^{n(\lam_1-\lam_2)}&1&-e^{n(\lam_3-\lam_2)}&-e^{n(\lam_4-\lam_2)}\\0&0&1&0\\0&0&0&1\ea\right)\mbox{ ͸¾µµÄϱßÇøÓò}
\ea
\right.
\ee

\par
Then (\ref{ran-TRHP}) and (\ref{ran-2S1}) imply that
\be
S_+(x)=S_-(x)J_S(x),\quad J_S(x)=\left(\ba{cccc}0&1&0&0\\-1&0&0&0\\0&0&1&0\\0&0&0&1\ea\right),\quad x\in [z_2,z_3].
\ee

\par
Similarly, consider a lens with vertices $-z_1,z_1$, that is contained in $U_2$ and set
\be\label{ran-2S2}
S(z)=\left\{
\ba{l}
T(z)\left(\ba{cccc}1&0&0&0\\0&1&0&0\\-e^{n(\lam_1-\lam_3)}&-e^{n(\lam_2-\lam_3)}&1&-e^{n(\lam_4-\lam_3)}\\0&0&0&1\ea\right)\mbox{ ͸¾µµÄÉϱßÇøÓò}\\
T(z)\left(\ba{cccc}1&0&0&0\\0&1&0&0\\e^{n(\lam_1-\lam_3)}&-e^{n(\lam_2-\lam_3)}&1&-e^{n(\lam_4-\lam_3)}\\0&0&0&1\ea\right)\mbox{ ͸¾µµÄϱßÇøÓò}
\ea
\right.
\ee

\par
Then (\ref{ran-TRHP}) and (\ref{ran-2S2}) imply that
\be
S_+(x)=S_-(x)J_S(x),\quad J_S(x)=\left(\ba{cccc}0&0&1&0\\0&1&0&0\\-1&0&0&0\\0&0&0&1\ea\right),\quad x\in [-z_1,z_1].
\ee

\par
Consider a lens with vertices $-z_3,-z_2$, that is contained in $U_3$ and set
\be\label{ran-2S3}
S(z)=\left\{
\ba{l}
T(z)\left(\ba{cccc}1&0&0&0\\0&1&0&0\\0&0&1&0\\-e^{n(\lam_1-\lam_4)}&-e^{n(\lam_2-\lam_4)}&-e^{n(\lam_3-\lam_4)}&1\ea\right)\mbox{ ͸¾µµÄÉϱßÇøÓò}\\
T(z)\left(\ba{cccc}1&0&0&0\\0&1&0&0\\0&0&1&0\\e^{n(\lam_1-\lam_4)}&-e^{n(\lam_2-\lam_4)}&-e^{n(\lam_3-\lam_4)}&1\ea\right)\mbox{ ͸¾µµÄϱßÇøÓò}
\ea
\right.
\ee
\par
Then (\ref{ran-TRHP}) and (\ref{ran-2S3}) imply that
\be
S_+(x)=S_-(x)J_S(x),\quad J_S(x)=\left(\ba{cccc}0&0&0&1\\0&1&0&0\\0&0&1&0\\-1&0&0&0\ea\right),\quad x\in [-z_3,-z_2].
\ee

\par
And set
\be\label{ran-2S}
S(z)=T(z)\qquad \mbox{outside of the lens region}.
\ee
\par
Then we have jumps on the boundary of the lenses,
\be\label{ran-SRHP}
S_+(z)=S_-(z)J_S(z)
\ee
where the contours are oriented from left to right (that is, from $-z_3$ to $-z_2$, or from $-z_1$ to
$z_1$, or from $z_2$ to $z_3$). The jump matrix in (\ref{ran-SRHP}) is
\small
\be\label{ran-Sjump}
\ba{l}
J_S(z)=\left(\ba{cccc}1&0&0&0\\e^{n(\lam_1-\lam_2)}&1&e^{n(\lam_3-\lam_2)}&e^{n(\lam_4-\lam_2)}\\0&0&1&0\\0&0&0&1\ea\right),\mbox{on the upper boundary of the $[z_2,z_3]-$lens}\\
J_S(z)=\left(\ba{cccc}1&0&0&0\\e^{n(\lam_1-\lam_2)}&1&-e^{n(\lam_3-\lam_2)}&-e^{n(\lam_4-\lam_2)}\\0&0&1&0\\0&0&0&1\ea\right),\mbox{on the lower boundary of the $[z_2,z_3]-$lens}\\
J_S(z)=\left(\ba{cccc}1&0&0&0\\0&1&0&0\\e^{n(\lam_1-\lam_3)}&e^{n(\lam_2-\lam_3)}&1&e^{n(\lam_4-\lam_3)}\\0&0&0&1\ea\right),\mbox{on the upper boundary of the $[-z_1,z_1]-$lens}\\
J_S(z)=\left(\ba{cccc}1&0&0&0\\0&1&0&0\\e^{n(\lam_1-\lam_3)}&-e^{n(\lam_2-\lam_3)}&1&-e^{n(\lam_4-\lam_3)}\\0&0&0&1\ea\right),\mbox{on the lower boundary of the $[-z_1,z_1]-$lens}\\
J_S(z)=\left(\ba{cccc}1&0&0&0\\0&1&0&0\\0&0&1&0\\e^{n(\lam_1-\lam_4)}&e^{n(\lam_2-\lam_4)}&e^{n(\lam_3-\lam_4)}&1\ea\right),\mbox{on the upper boundary of the $[-z_3,-z_2]-$lens}\\
J_S(z)=\left(\ba{cccc}1&0&0&0\\0&1&0&0\\0&0&1&0\\e^{n(\lam_1-\lam_4)}&-e^{n(\lam_2-\lam_4)}&-e^{n(\lam_3-\lam_4)}&1\ea\right),\mbox{on the lower boundary of the $[-z_3,-z_2]-$lens}
\ea
\ee
\normalsize
And
\be\label{ran-Sjumpreal}
S_+(x)=S_-(x)J_S(x),\qquad J_S(x)=J_T(x),\quad x\in (-\infty,-z_3]\cup [-z_2,-z_1]\cup [z_1,z_2]\cup [z_3,\infty).
\ee

\par

Thus $S$ solves the Riemann-Hilbert problem
\begin{itemize}
 \item $S$ is analytic on $\C\backslash (\R\cup \Gam)$, where $\Gam$ denotes the boundary of the lens,

 \item
  \be\label{ran-TRHP}
  S_+(z)=S_-(z)J_S(z),\quad z\in \R\cup \Gam,
  \ee

 \item As $z\ra \infty$,
 \be
 S(z)=\id+O(\frac{1}{z}).
 \ee
\end{itemize}

The kernel $K_n(x,y)$ is expressed in terms of $S$ as follows. By (\ref{ran-KnT}) and the definitions (\ref{ran-2S1}), (\ref{ran-2S2}), (\ref{ran-2S3}), for $x,y$ in $(z_2,z_3)$, we have
\be\label{ran-KnS1}
K_n(x,y)=\frac{e^{\frac{n}{4}(x^2-y^2)}}{2\pi i(x-y)}\left(\ba{cccc}-e^{n\lam_{1+}(y)}&e^{n\lam_{2+}(y)}&0&0\ea\right)S_{+}^{-1}(y)S_+(x)\left(\ba{c}e^{-n\lam_{1+}(x)}\\e^{-n\lam_{2+}(x)}\\0\\0\ea\right),
\ee
for $x,y$ in $(-z_1,z_1)$, we have
\be\label{ran-KnS2}
K_n(x,y)=\frac{e^{\frac{n}{4}(x^2-y^2)}}{2\pi i(x-y)}\left(\ba{cccc}-e^{n\lam_{1+}(y)}&0&e^{n\lam_{3+}(y)}&0\ea\right)S_{+}^{-1}(y)S_+(x)\left(\ba{c}e^{-n\lam_{1+}(x)}\\0\\e^{-n\lam_{3+}(x)}\\0\ea\right),
\ee
while for $x,y$ in $(-z_3,-z_2)$, we have
\be\label{ran-KnS3}
K_n(x,y)=\frac{e^{\frac{n}{4}(x^2-y^2)}}{2\pi i(x-y)}\left(\ba{cccc}-e^{n\lam_{1+}(y)}&0&0&e^{n\lam_{4+}(y)}\ea\right)S_{+}^{-1}(y)S_+(x)\left(\ba{c}e^{-n\lam_{1+}(x)}\\0\\0\\e^{-n\lam_{4+}(x)}\ea\right).
\ee

\par
Since $\lam_{1+}$ and $\lam_{2+}$ are complex conjugates on $(z_2,z_3)$, we can rewrite (\ref{ran-KnS1}) for $x,y$ in $(z_2,z_3)$ as
\be\label{ran-KnS11}
K_n(x,y)=\frac{e^{n(h(y)-h(x))}}{2\pi i(x-y)}\left(\ba{cccc}-e^{ni\im\lam_{1+}(y)}&e^{-ni\im\lam_{1+}(y)}&0&0\ea\right)S_{+}^{-1}(y)S_+(x)\left(\ba{c}e^{-ni\im\lam_{1+}(x)}\\e^{ni\im\lam_{1+}(x)}\\0\\0\ea\right),
\ee
where $h(x)=-\frac{1}{4}x^2+\re \lam_{1+}(x)$.
Similarly, for $x,y$ in $(-z_1,z_1)$, we have
\be\label{ran-KnS22}
K_n(x,y)=\frac{e^{n(h(y)-h(x))}}{2\pi i(x-y)}\left(\ba{cccc}-e^{ni\im\lam_{1+}(y)}&0&e^{-ni\im\lam_{1+}(y)}&0\ea\right)S_{+}^{-1}(y)S_+(x)\left(\ba{c}e^{-ni\im\lam_{1+}(x)}\\0\\e^{ni\im\lam_{1+}(x)}\\0\ea\right),
\ee
for $x,y$ in $(-z_3,-z_2)$, we have
\be\label{ran-KnS33}
K_n(x,y)=\frac{e^{n(h(y)-h(x))}}{2\pi i(x-y)}\left(\ba{cccc}-e^{ni\im\lam_{1+}(y)}&0&0&e^{-ni\im\lam_{1+}(y)}\ea\right)S_{+}^{-1}(y)S_+(x)\left(\ba{c}e^{-ni\im\lam_{1+}(x)}\\0\\0\\e^{ni\im\lam_{1+}(x)}\ea\right).
\ee

\subsection{Model Riemann-Hilbert problem}
As $n\ra \infty$, the jump matrix $J_S(z)$ is exponentially close to the identity matrix at every $z$ outside of $[-z_3,-z_2]\cup [-z_1,z_1]\cup [z_2,z_3]$. This follows from (\ref{ran-Sjump}) and Lemma \ref{lamdaxiao} for $z$ on the boundary of the lenses, and (\ref{ran-Sjumpreal}), (\ref{ran-Tjumprealsheng}), Lemma \ref{lamdaxiao} for $z$ on the intervals $(-\infty,-z_3]\cup [-z_2,-z_1]\cup [z_1,z_2]\cup [z_3,\infty)$.

\par
In this subsection, we solve the following model Riemann-Hilbert problem, where we ignore the
exponentially small jumps: find $M:\C\backslash ([-z_3,-z_2]\cup [-z_1,z_1]\cup [z_2,z_3])\ra \C^{4\times 4}$ such that
\begin{itemize}
 \item $M$ is analytic on $\C\backslash ([-z_3,-z_2]\cup [-z_1,z_1]\cup [z_2,z_3])$,

 \item
  \be\label{ran-ModelRHP}
  M_+(x)=M_-(x)J_S(x),\quad x\in (-z_3,-z_2)\cup (-z_1,z_1)\cup (z_2,z_3),
  \ee

 \item As $z\ra \infty$,
 \be\label{ran-Modelasy}
 M(z)=\id+O(\frac{1}{z})
 \ee
\end{itemize}
This problem is similar to the RH problem considered in [22, Sect. 6.1].We also follow
a similar method to solve it.

\par
We lift the model Riemann-Hilbert problem to the Riemann surface of (\ref{Specurve}) with the sheet structure
as in figure \ref{ran-fig1}. Consider
\be
\ba{l}
\Omega_1=\xi_1(\C\backslash ([-z_3,-z_2]\cup [-z_1,z_1]\cup [z_2,z_3])),\\
\Omega_2=\xi_2(\C\backslash [z_2,z_3]),\\
\Omega_3=\xi_3(\C\backslash [-z_1,z_1]),\\
\Omega_4=\xi_4(\C\backslash [-z_3,-z_2]).
\ea
\ee
Then $\Omega_1,\Omega_2,\Omega_3,\Omega_4$ give a partition of the complex plane into four regions, see figure \ref{ran-fig3}.
\begin{figure}[th]
\centering
\includegraphics{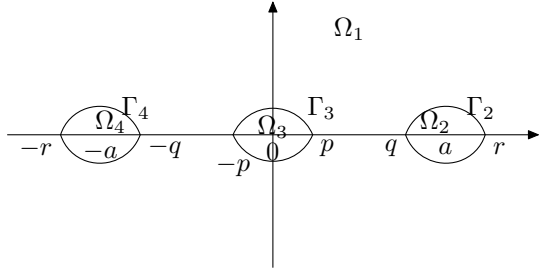}
\caption{Partition of the complex $\xi$ plane}\label{ran-fig3}
\end{figure}
The regions $\Omega_2,\Omega_3,\Omega_4$ are bounded, $a\in \Omega_2$, $0\in \Omega_3$, $-a\in\Omega_4$, with the symmetry conditions
\be
\bar\Omega_2=\Omega_2,\quad \bar \Omega_4=\Omega_4,\quad \bar \Omega_3=\Omega_3,\quad \Omega_2=-\Omega_4,\quad \Omega_3=-\Omega_3.
\ee

\par
Denote by $\Gam_k$ the boundary of $\Omega_k,k=2,3,4$, see figure \ref{ran-fig3}. Then we have
\be
\ba{l}
\xi_{1+}([z_2,z_3])=\xi_{2-}([z_2,z_3])=\Gam^{+}_2=\Gam_2\cap \{\im z\ge 0\},\\
\xi_{1-}([z_2,z_3])=\xi_{2+}([z_2,z_3])=\Gam^{-}_2=\Gam_2\cap \{\im z\le 0\},\\
\xi_{1+}([-z_1,z_1])=\xi_{3-}([-z_1,z_1])=\Gam^{+}_3=\Gam_3\cap \{\im z\ge 0\},\\
\xi_{1-}([-z_1,z_1])=\xi_{3+}([-z_1,z_1])=\Gam^{-}_3=\Gam_3\cap \{\im z\le 0\},\\
\xi_{1+}([-z_3,-z_2])=\xi_{4-}([-z_3,-z_2])=\Gam^{+}_4=\Gam_4\cap \{\im z\ge 0\},\\
\xi_{1-}([-z_3,-z_2])=\xi_{4+}([-z_3,-z_2])=\Gam^{-}_4=\Gam_4\cap \{\im z\le 0\}.
\ea
\ee

\par
We are looking for a solution $M$ in the following form:
\be
M(z)=\left(\ba{cccc}
              M_{1}(\xi_1(z))&M_1(\xi_2(z))&M_1(\xi_3(z))&M_1(\xi_4(z))\\
               M_{2}(\xi_1(z))&M_2(\xi_2(z))&M_2(\xi_3(z))&M_2(\xi_4(z))\\
                M_{3}(\xi_1(z))&M_3(\xi_2(z))&M_3(\xi_3(z))&M_3(\xi_4(z))\\
                 M_{4}(\xi_1(z))&M_4(\xi_2(z))&M_4(\xi_3(z))&M_4(\xi_4(z))\\
              \ea\right),
\ee
where $M_1(\xi),M_{2}(\xi),M_3(\xi),M_4(\xi)$ are four analytic functions on $\C\backslash (\Gam_2\cup\Gam_3\cup\Gam_4)$. To satisfy
the jump condition (\ref{ran-ModelRHP}), we need the following relations for $k=1,2,3,4$,
\be\label{ran-ModelMk}
\ba{l}
M_{k+}(\xi)=-M_{k-}(\xi),\qquad \xi\in \Gam_2^{+}\cup \Gam_3^{+}\cup \Gam_4^{+},\\
M_{k+}(\xi)=M_{k-}(\xi),\qquad \xi\in \Gam_2^{-}\cup \Gam_3^{-}\cup \Gam_4^{-}.
\ea
\ee
Since $\xi_1(\infty)=\infty,\xi_2(\infty)=a,\xi_3(\infty)=0,\xi_4(\infty)=-a$, then to satisfy (\ref{ran-Modelasy}), we have
\be\label{ran-ModelMdian}
\ba{llll}
M_1(\infty)=1,&M_1(a)=0,&M_1(0)=0,&M_1(-a)=0,\\
M_2(\infty)=0,&M_2(a)=1,&M_2(0)=0,&M_2(-a)=0,\\
M_3(\infty)=0,&M_3(a)=0,&M_3(0)=1,&M_4(-a)=0,\\
M_4(\infty)=0,&M_4(a)=0,&M_3(0)=0,&M_4(-a)=1.
\ea
\ee

\par
Equation (\ref{ran-ModelMk}) and (\ref{ran-ModelMdian}) have the following solution
\be
\ba{ll}
M_1(\xi)=\frac{\xi(\xi^2-a^2)}{\sqrt{(\xi^2-p^2)(\xi^2-q^2)(\xi^2-r^2)}},&M_2(\xi)=c_2\frac{\xi(\xi+a)}{\sqrt{(\xi^2-p^2)(\xi^2-q^2)(\xi^2-r^2)}},\\
M_3(\xi)=c_3\frac{\xi^2-a^2}{\sqrt{(\xi^2-p^2)(\xi^2-q^2)(\xi^2-r^2)}},&M_4(\xi)=c_4\frac{\xi(\xi-a)}{\sqrt{(\xi^2-p^2)(\xi^2-q^2)(\xi^2-r^2)}}.
\ea
\ee
Here the constants $c_2,c_3,c_4$ are defined by $M_2(a)=1,M_3(0)=1,M_4(-a)=1$. Notice
\be
(\xi^2-p^2)(\xi^2-q^2)(\xi^2-r^2)=\xi^6-(1+2a^2)\xi^4+[a^4+(3t-1)a^2]\xi^2-ta^4,
\ee
hence
\be
M_2(a)=c_2\frac{2a^2}{\sqrt{2(t-1)a^4}}.
\ee
By taking into account the cuts of $M_2(\xi)$, we have
\be
M_2(a)=c_2i\sqrt{\frac{2}{1-t}},
\ee
hence
\be
c_2=-i\sqrt{\frac{1-t}{2}}.
\ee
Similarly, we have
\be
c_3=-i\sqrt{t},\qquad c_4=-i\sqrt{\frac{1-t}{2}}.
\ee

Thus, the solution to the model Riemann-Hilbert is given as following,
\tiny
\be
\ba{l}
M(z)=\\
\left(\ba{cccc}
              \frac{\xi_1(\xi_1^2-a^2)}{\sqrt{(\xi_1^2-p^2)(\xi_1^2-q^2)(\xi_1^2-r^2)}}&\frac{\xi_2(\xi_2^2-a^2)}{\sqrt{(\xi_2^2-p^2)(\xi_2^2-q^2)(\xi_2^2-r^2)}}&\frac{\xi_3(\xi_3^2-a^2)}{\sqrt{(\xi_3^2-p^2)(\xi_3^2-q^2)(\xi_3^2-r^2)}}&\frac{\xi_4(\xi_4^2-a^2)}{\sqrt{(\xi_4^2-p^2)(\xi_4^2-q^2)(\xi_4^2-r^2)}}\\
               -i\frac{\sqrt{1-t}\xi_1(\xi_1+a)}{\sqrt{2(\xi_1^2-p^2)(\xi_1^2-q^2)(\xi_1^2-r^2)}}&-i\frac{\sqrt{1-t}\xi_2(\xi_2+a)}{\sqrt{2(\xi_2^2-p^2)(\xi_2^2-q^2)(\xi_2^2-r^2)}}&-i\frac{\sqrt{1-t}\xi_3(\xi_3+a)}{\sqrt{2(\xi_3^2-p^2)(\xi_3^2-q^2)(\xi_3^2-r^2)}}&-i\frac{\sqrt{1-t}\xi_4(\xi_4+a)}{\sqrt{2(\xi_4^2-p^2)(\xi_4^2-q^2)(\xi_4^2-r^2)}}\\
                -i\frac{\sqrt{t}(\xi_1^2-a^2)}{\sqrt{(\xi_1^2-p^2)(\xi_1^2-q^2)(\xi_1^2-r^2)}}&-i\frac{\sqrt{t}(\xi_2^2-a^2)}{\sqrt{(\xi_2^2-p^2)(\xi_2^2-q^2)(\xi_2^2-r^2)}}&-i\frac{\sqrt{t}(\xi_3^2-a^2)}{\sqrt{(\xi_3^2-p^2)(\xi_3^2-q^2)(\xi_3^2-r^2)}}&-i\frac{\sqrt{t}(\xi_4^2-a^2)}{\sqrt{(\xi_4^2-p^2)(\xi_4^2-q^2)(\xi_4^2-r^2)}}\\
                 -i\frac{\sqrt{1-t}\xi_1(\xi_1-a)}{\sqrt{2(\xi_1^2-p^2)(\xi_1^2-q^2)(\xi_1^2-r^2)}}&-i\frac{\sqrt{1-t}\xi_2(\xi_2-a)}{\sqrt{2(\xi_2^2-p^2)(\xi_2^2-q^2)(\xi_2^2-r^2)}}&-i\frac{\sqrt{1-t}\xi_3(\xi_3-a)}{\sqrt{2(\xi_3^2-p^2)(\xi_3^2-q^2)(\xi_3^2-r^2)}}&-i\frac{\sqrt{1-t}\xi_4(\xi_4-a)}{\sqrt{2(\xi_4^2-p^2)(\xi_4^2-q^2)(\xi_4^2-r^2)}}\\
              \ea\right)
\ea
\ee
\normalsize

\par
The solution $M$ to the model Riemann-Hilbert problem will be used to construct a parametrix for the Riemann-Hilbert problem
for $S$ outside of a small neighborhood of the edge points. Namely, we will fix some $r>0$ and consider the disks of radius $r$ around the edge points. At the edge points $M$ is not analytic and in a neighborhood of the edge points the parametrix is constructed
differently.

\subsection{Parametrix at Edge Points}

We consider small disks $D(\pm z_j,r),j=1,2,3$ with radius $r>0$ and centered at the edge points,
and look for a local parametrix $P$ defined on the union of the six disks such that

\begin{itemize}
  \item $P$ is analytic at $D(\pm z_j,r)\backslash(\R\cup \Gam)$,

  \item
   \be
   P_+(z)=P_-(z)J_S(z),\qquad z\in (\R\cup \Gam)\cap D(\pm z_j,r),
   \ee

 \item As $n\ra \infty$,
  \be\label{ran-Pasy}
  P(z)=\left(\id+O(\frac{1}{z})\right)M(z),\quad \mbox{uniformly for $z\in \partial D(\pm z_j,r)$}.
  \ee
\end{itemize}

We consider here the edge point $z_3$ in detail. We note that as $z\ra z_3$,
\be
\ba{l}
\lam_1(z)=r(z-z_3)+\frac{2\rho_1}{3}(z-z_3)^{\frac{3}{2}}+O(z-z_3)^2,\\
\lam_2(z)=r(z-z_3)-\frac{2\rho_1}{3}(z-z_3)^{\frac{3}{2}}+O(z-z_3)^2,
\ea
\ee
thus, we have as $z\ra z_3$,
\be\label{ran-lamda12}
\lam_1(z)-\lam_2(z)=\frac{4\rho_1}{3}(z-z_3)^{\frac{3}{2}}+O(z-z_3)^{\frac{5}{2}}.
\ee
Hence
\be
\beta(z)=\left[\frac{3}{4}(\lam_1(z)-\lam_2(z))\right]^{\frac{2}{3}}
\ee
is analytic at $z_3$, real-valued on the real axis near $z_3$, and $\beta'(z)=\rho_1^{\frac{2}{3}}>0$. So $\beta$ is a
conformal map from $D(z_3,r)$ to a convex neighborhood of the origin, if $r$ is sufficiently
small. We take $\Gam$ near $z_3$ such that
\be
\beta(\Gam\cap D(z_3,r))\subset \{z|arg(z)=\pm \frac{2\pi}{3}\}.
\ee
Thus, $\Gam$ and $\R$ divide the disk $D(z_3,r)$ into four regions numbered $\Rmnum{1},\Rmnum{2},\Rmnum{3}$ and $\Rmnum{4}$ such that $0<arg \beta(z)<\frac{2\pi}{3},\frac{2\pi}{3}<\arg\beta(z)<\pi,-\pi<arg\beta(z)<-\frac{2\pi}{3}$ and $-\frac{2\pi}{3}<arg\beta(z)<0$ for $z$ in regions $\Rmnum{1},\Rmnum{2},\Rmnum{3}$ and $\Rmnum{4}$, respectively.

\par
Recall that the jumps $J_S$ near $z_3$ are given by
\be
\ba{l}
J_S=\left(\ba{cccc}0&1&0&0\\-1&0&0&0\\0&0&1&0\\0&0&0&1\ea\right),\quad x\in [z_3-r,z_3),\\
J_S=\left(\ba{cccc}1&0&0&0\\e^{n(\lam_1-\lam_2)}&1&e^{n(\lam_3-\lam_2)}&e^{n(\lam_4-\lam_2)}\\0&0&1&0\\0&0&0&1\ea\right),\mbox{on the upper boundary of the lens in $D(z_3,r)$}\\
J_S=\left(\ba{cccc}1&0&0&0\\e^{n(\lam_1-\lam_2)}&1&-e^{n(\lam_3-\lam_2)}&-e^{n(\lam_4-\lam_2)}\\0&0&1&0\\0&0&0&1\ea\right),\mbox{on the lower boundary of the lens in $D(z_3,r)$}\\
J_S=\left(
          \ba{cccc}
          1&e^{n(\lam_{2}-\lam_{1})}&e^{n(\lam_{3}-\lam_{1})}&e^{n(\lam_{4}-\lam_{1})}\\
          0&1&0&0\\
          0&0&1&0\\
          0&0&0&1
          \ea
          \right),\quad x\in (z_3,z_3+r].
\ea
\ee

We write
\be
\tilde P=\left\{
\ba{ll}
P\left(\ba{cccc}1&0&0&0\\0&1&-e^{n(\lam_3-\lam_2)}&-e^{n(\lam_4-\lam_2)}\\0&0&1&0\\0&0&0&1\ea\right),&\mbox{in regions $\Rmnum{1},\Rmnum{4}$},\\
P,&\mbox{in regions $\Rmnum{2},\Rmnum{3}$}.
\ea
\right.
\ee

Then the jumps $\tilde P$ are $\tilde P_+=\tilde P_-J_{\tilde P}$, where
\be\label{ran-tildePjump}
\ba{l}
J_{\tilde P}=\left(\ba{cccc}0&1&0&0\\-1&0&0&0\\0&0&1&0\\0&0&0&1\ea\right),\quad x\in [z_3-r,z_3),\\
J_{\tilde P}=\left(\ba{cccc}1&0&0&0\\e^{n(\lam_1-\lam_2)}&1&0&0\\0&0&1&0\\0&0&0&1\ea\right),\mbox{on the upper boundary of the lens in $D(z_3,r)$}\\
J_{\tilde P}=\left(\ba{cccc}1&0&0&0\\e^{n(\lam_1-\lam_2)}&1&0&0\\0&0&1&0\\0&0&0&1\ea\right),\mbox{on the lower boundary of the lens in $D(z_3,r)$}\\
J_{\tilde P}=\left(
          \ba{cccc}
          1&e^{n(\lam_{2}-\lam_{1})}&0&0\\
          0&1&0&0\\
          0&0&1&0\\
          0&0&0&1
          \ea
          \right),\quad x\in (z_3,z_3+r].
\ea
\ee

\par
The Riemann-Hilbert problem for $\tilde P$ is essentially a $2\times 2$ problem, since the jumps (\ref{ran-tildePjump}) are
non-trivial only in the upper $2\times 2$ block. A solution can be constructed in a standard
way out of Airy functions. The Airy function $Ai(z)$ solves the differential equation
\be\label{ran-airy}
y^{''}=zy.
\ee

For any $\eps>0$, in the sector $\pi+\eps\le arg z\le \pi-\eps$, it has the asymptotics as $z\ra\infty$,
\be
Ai(z)=\frac{1}{2\sqrt{\pi}z^{\frac{1}{4}}}e^{-\frac{2}{3}z^{\frac{3}{2}}}\left(1+O(z^{-\frac{3}{2}})\right).
\ee

The functions $Ai(\omega z),Ai(\omega^2z)$, where $\omega=e^{\frac{2\pi i}{3}}$, also solve the equation (\ref{ran-airy}). And we have the linear relation,
\be
Ai(z)+\omega Ai(\omega z)+\omega^2Ai(\omega^2z)=0.
\ee

\par
Write
\be
y_0(z)=Ai(z),\quad y_1(z)=\omega Ai(\omega z),\quad y_2(z)=\omega^2Ai(\omega^2z),
\ee
using these functions to define
\be\label{ran-Phidef}
\phi(z)=\left\{
\ba{ll}
\left(\ba{cccc}y_0(z)&-y_2(z)&0&0\\y'_0(z)&-y'_2(z)&0&0\\0&0&1&0\\0&0&0&1\ea\right)&\mbox{µ±~$0<arg z<\frac{2\pi}{3}$},\\
\left(\ba{cccc}-y_1(z)&-y_2(z)&0&0\\-y'_1(z)&-y'_2(z)&0&0\\0&0&1&0\\0&0&0&1\ea\right)&\mbox{µ±~$\frac{2\pi}{3}<arg z<\pi$},\\
\left(\ba{cccc}-y_2(z)&y_1(z)&0&0\\-y'_2(z)&y'_1(z)&0&0\\0&0&1&0\\0&0&0&1\ea\right)&\mbox{µ±~$-\pi<arg z<-\frac{2\pi}{3}$},\\
\left(\ba{cccc}y_0(z)&y_1(z)&0&0\\y'_0(z)&y'_1(z)&0&0\\0&0&1&0\\0&0&0&1\ea\right)&\mbox{µ±~$-\frac{2\pi}{3}<arg z<0$}.
\ea
\right.
\ee

\par
To match the asymptotic condition as $z\ra\infty$, we should have
\be
\tilde P=\left(\id+O(\frac{1}{n})\right)M(z),\quad \mbox{uniformly in $z\in \partial D(z_3,r)$}.
\ee

Then
\be\label{ran-tildeP}
\tilde P(z)=E_n(z)\phi(n^{\frac{2}{3}}\beta(z))diag\left(e^{\frac{1}{2}n(\lam_1(z)-\lam_2(z))},e^{-\frac{1}{2}n(\lam_1(z)-\lam_2(z))},1,1\right),
\ee
where
\be
E_n=\sqrt{\pi}M\left(\ba{cccc}1&-1&0&0\\-i&-i&0&0\\0&0&1&0\\0&0&0&1\ea\right)\left(\ba{cccc}n^{\frac{1}{6}}\beta^{\frac{1}{4}}&0&0&0\\0&n^{-\frac{1}{6}}\beta^{-\frac{1}{4}}&0&0\\0&0&1&0\\0&0&0&1\ea\right).
\ee

\par
A similar construction works for a parametrix $P$ around the other edge points.

\subsection{Third Transformation}
Let
\be\label{ran-Rbianhuan}
\ba{ll}
R(z)=S(z)M(z)^{-1}&\mbox{$z$ outside of the disks $D(\pm z_j,r),j=1,2,3$},\\
R(z)=S(z)P(z)^{-1}&\mbox{$z$ inside of the disks}.
\ea
\ee
Then we have $R(z)$ is analytic on $\C\backslash \Gam_R$, where $\Gam_R$ consists of the six circles $\partial D(\pm z_j,r),j=1,2,3$,
the parts of $\Gam$ outside of the six disks, and the real intervals $(-\infty,-z_3-r),(-z_2+r,-z_1-r),(z_1+r,z_2-r),(z_3+r,\infty)$.

\par
The jump relations for $R(z)$ are
\be\label{ran-RRHP}
R_+=R_-J_R,
\ee
where
\be
\ba{ll}
J_R=MP^{-1},&\mbox{on the circles, oriented counterclockwise},\\
J_R=MJ_SM^{-1},&\mbox{on the remaining parts of $\Gam_R$}.
\ea
\ee

From (\ref{ran-Pasy}) it follows that $J_R=\id+O(\frac{1}{n})$ uniformly on the circles, and from (\ref{ran-Sjump}), (\ref{ran-Sjumpreal}), (\ref{ran-Tjump})
and lemma \ref{lamdaxiao} it follows that $J_R=\id+O(e^{-cn})$ for $c>0$ as $n\ra\infty$, uniformly on the remaining parts of $\Gam_R$. So we can conclude
\be\label{ran-Rasyn}
J_R(z)=\id+O(\frac{1}{n}),\quad \mbox{as $n\ra\infty$, uniformly on $\Gam_R$}.
\ee

As $z\ra\infty$, we have
\be\label{ran-Rasy}
R(z)=\id+O(\frac{1}{z}).
\ee

\par
From (\ref{ran-RRHP}), (\ref{ran-Rasyn}), (\ref{ran-Rasy}) and and the fact that we can deform the contours in any desired
direction, it follows that
\be
R(z)=\id+O(\frac{1}{n|z|+1}),\qquad n\ra\infty,\quad \mbox{uniformly for $z \in \R\backslash \Gam_R$},
\ee
see \cite{deiftbook}.

\par
By Cauchy theorem, we then also have
\be
R'(z)=O(\frac{1}{n|z|+1}),
\ee
thus,
\be\label{ran-Rjiexi}
R^{-1}(y)R(x)=\id+R^{-1}(y)(R(x)-R(y))=\id+O(\frac{x-y}{n}).
\ee

\section{Proofs of the Results}
\begin{theorem}
Let $V(M)=\frac{1}{2}M^2$, $n_1=n_3$, $n_2=t$, $a^2>3$. Then
the limiting mean density of eigenvalues
\be
\rho(z)=\lim_{n\ra\infty}\frac{1}{n}K_n(x,x),
\ee
exists, and it is supported by three intervals, $[-z_3,-z_2], [-z_1,z_1]$ and $[z_2,z_3]$.
\end{theorem}
\begin{proof}
Consider $x\in(z_2,z_3)$. We may assume that the circles
around the edge points are such that $x$ is outside of the six disks. From (\ref{ran-Rbianhuan}) it follows that $S(x)=R(x)M(x)$. And from (\ref{ran-Rjiexi})and $M_+$ is real analytic in a neighborhood of $x$, then
\be
S^{-1}_+(y)S_+(x)=\id+O(x-y),\quad \mbox{µ± }\quad y\ra x,
\ee
uniformly for all $n$. Thus, we have
\be
\ba{rl}
K_n(x,y)&=\frac{e^{n(h(y)-h(x))}}{2\pi i(x-y)}\left(\ba{cccc}-e^{ni\im\lam_{1+}(y)}&e^{-ni\im\lam_{1+}(y)}&0&0\ea\right)(\id+O(x-y))\left(\ba{c}e^{-ni\im\lam_{1+}(x)}\\e^{ni\im\lam_{1+}(x)}\\0\\0\ea\right)\\
{}&=e^{n(h(y)-h(x))}\left(\frac{-e^{ni(\im\lam_{1+}(y)-\im\lam_{1+}(x))}+e^{-ni(\im\lam_{1+}(y)-\im\lam_{1+}(x))}}{2\pi i(x-y)}+O(1)\right)\\
{}&=e^{n(h(y)-h(x))}\left(\frac{sin{(n\im(\lam_{1+}(x)-\lam_{1+}(y)))}}{\pi (x-y)}+O(1)\right),
\ea
\ee
Letting $y\ra x$ and noting that
\be
\frac{d}{dy}\im\lam_{1+}(y)=\im\xi_{1+}(y)=\pi \rho(y),
\ee
By L'Hospital rule, we have
\be
K_n(x,x)=n\rho(x)+O(1).
\ee
Similarly, for $x$ in $(-z_3,-z_2), (-z_1,z_1)$

For $x\in (-\infty,-z_3)\cup (-z_2,-z_1)\cup (z_1,z_2)\cup (z_3,\infty)$, we have $K_n(x,x)$ decreases exponentially
fast. This implies that
\be
\lim_{n\ra\infty}\frac{1}{n}K_n(x,x)=0.
\ee

If x is one of the edge points, it is shown in the proof of \ref{ran-thm3} that as $n\ra \infty$,
\be
\frac{1}{n}K_n(x,x)=O\left(\frac{1}{n^{\frac{1}{3}}}\right).
\ee

\end{proof}

Let
\be
\hat K_n(x,y)=e^{n(h(x)-h(y))}K_n(x,y),
\ee
where for $x\in (-z_3,-z_2)\cup (-z_1,z_1)\cup (z_2,z_3)$,
\be
h(x)=-\frac{1}{4}x^2+\re \lam_{1+}(x),\quad \lam_{1+}(x)=\int_{z_1}^x\xi_{1+}(s)ds.
\ee
\begin{theorem}\label{ran-thm2}
For every $x_0\in(-z_3,-z_2)\cup (-z_1,z_1)\cup (z_2,z_3)$ and $u,v\in \R$, we have
\be
\lim_{n\ra\infty}\frac{1}{n\rho(x_0)}\hat K_n\left(x_0+\frac{u}{n\rho(x_0)},x_0+\frac{v}{n\rho(x_0)}\right)=\frac{sin \pi(u-v)}{\pi(u-v)}.
\ee
\end{theorem}
\begin{proof}
We just prove $x\in(z_2,z_3)$, similarly for the rest. Letting
\be
x=x_0+\frac{u}{n\rho(x_0)},\quad y=x_0+\frac{v}{n\rho(x_0)}.
\ee
By the definition of $\hat K_n$, we have
\be
\frac{1}{n\rho(x_0)}\hat K_n(x,y)=\frac{sin(n\im(\lam_{1+}(x)-\lam_{1+}(y)))}{\pi (u-v)}+O(\frac{1}{n}).
\ee
By the mean value theorem, we have
\be
\im(\lam_{1+}(x)-\lam_{1+}(y))=(x-y)\pi \rho(t),
\ee
for some $t$ between $x$ and $y$. We also can know that $t=x_0+O(1/n)$ and
\be
n\im(\lam_{1+}(x)-\lam_{1+}(y))=\pi(u-v)\frac{\rho(t)}{\rho(x_0)}=\pi(u-v)\left(1+O(\frac{1}{n})\right).
\ee
Then we have
\be
\frac{1}{n\rho(x_0)}\hat K_n(x,y)=\frac{sin\pi(u-v)}{\pi (u-v)}+O(\frac{1}{n}).
\ee
This is proven \ref{ran-thm2}¡£
\end{proof}

\begin{theorem}\label{ran-thm3}
We use the same notation as in Theorem \ref{ran-thm2}.
For every $u,v\in \R$, we have
\be
\lim_{n\ra\infty}\frac{1}{(\rho_1n)^{\frac{2}{3}}}\hat K_n\left(z_3+\frac{u}{(\rho_1n)^{\frac{2}{3}}},z_3+\frac{v}{(\rho_1n)^{\frac{2}{3}}}\right)=\frac{Ai(u)Ai'(v)-Ai'(u)Ai(v)}{u-v},
\ee
\be
\lim_{n\ra\infty}\frac{1}{(\rho_2n)^{\frac{2}{3}}}\hat K_n\left(z_2+\frac{u}{(\rho_2n)^{\frac{2}{3}}},z_2+\frac{v}{(\rho_2n)^{\frac{2}{3}}}\right)=\frac{Ai(u)Ai'(v)-Ai'(u)Ai(v)}{u-v},
\ee
\be
\lim_{n\ra\infty}\frac{1}{(\rho_3n)^{\frac{2}{3}}}\hat K_n\left(z_1+\frac{u}{(\rho_3n)^{\frac{2}{3}}},z_1+\frac{v}{(\rho_3n)^{\frac{2}{3}}}\right)=\frac{Ai(u)Ai'(v)-Ai'(u)Ai(v)}{u-v},
\ee
Similar limits hold near the edge points $-z_3,-z_2,-z_1$.
\end{theorem}
\begin{proof}
We just prove the first formula, similarly for the rest. Noting that $\beta'(z_3)=\rho_1^{\frac{2}{3}}$.

Choosing that $u,v\in\R$ and letting
\be
x=z_3+\frac{u}{(\rho_1n)^{\frac{2}{3}}},\quad y=z_3+\frac{v}{(\rho_1n)^{\frac{2}{3}}}
\ee
Suppose that $u,v<0$, so we can use the formula (\ref{ran-KnS11}) for $K_n(x,y)$. Thus, we have as $n$ goes to infinity, $x$ is inside of $D(z_3,r)$, from (\ref{ran-lamda12}), (\ref{ran-tildeP}), (\ref{ran-Rbianhuan}), we have
\be
\ba{rl}
S_+(x)=&R(x)P_+(x)=R(x)\tilde P(x)\\
{}=&R(x)E_n(x)\Phi_+(n^{\frac{2}{3}}\beta(x))diag\left(e^{\frac{1}{2}n(\lam_1(z)-\lam_2(z))},e^{-\frac{1}{2}n(\lam_1(z)-\lam_2(z))},1,1\right)\\
{}=&R(x)E_n(x)\Phi_+(n^{\frac{2}{3}}\beta(x))diag\left(e^{ni\im\lam_{1+}(x)},e^{-ni\im\lam_{1+}(x)},1,1\right),
\ea
\ee
$S_+(y)$ has similar equation hold. Thus,
\be
\ba{rl}
\frac{1}{(\rho_1n)^{2/3}}\hat K_n(x,y)=&\frac{1}{2\pi i(u-v)}\left(\ba{cccc}-1&1&0&0\ea\right)\Phi^{-1}_+(n^{2/3}\beta(y))E^{-1}_n(y)R^{-1}(y)\\
{}&{}\times R(x)E_n(x)\Phi_+(n^{2/3}\beta(x))\left(\ba{c}1\\1\\0\\0\ea\right).
\ea
\ee

Since $\rho^{2/3}=\beta'(z_3)$, as $n\ra\infty$,
\be
n^{2/3}\beta(x)=n^{2/3}\beta\left(z_3+\frac{u}{(\rho_1n)^{2/3}}\right)\ra u,
\ee
Then $\Phi_+(n^{2/3}\beta(x))\ra \Phi_+(u)$. We use the second formula of (\ref{ran-Phidef}) to compute $\Phi_+(u)$, we have
\be
\lim_{n\ra\infty}\Phi_+(n^{2/3}\beta(x))\left(\ba{c}1\\1\\0\\0\ea\right)=\left(-y_1(u)-y_2(u)\\-y'_1(u)-y'_2(u)\\0\\0\right)=\left(\ba{c}y_0(u)\\y'_0(u)\\0\\0\ea\right).
\ee
Similary, we have
\be
\lim_{n\ra\infty}\left(\ba{cccc}-1&1&0&0\ea\right)\Phi^{-1}_+(n^{2/3}\beta(y))=-2\pi i\left(\ba{cccc}-y'_0(v)&y_0(v)&0&0\ea\right).
\ee

Noting that $R^{-1}(y)R(x)=\id+O(\frac{x-y}{n})$, then
\be
R^{-1}(y)R(x)=\id+O(\frac{1}{n^{5/3}}).
\ee
From the explicit formula of $E_n$,
\be
E_n(x)=O(n^{1/6}),\quad E^{-1}_n(y)=O(n^{1/6}),\quad E^{-1}_n(y)E_n(x)=\id+O(\frac{1}{n^{1/3}}).
\ee
we have
\be
\lim_{n\ra\infty}E^{-1}_n(y)R^{-1}(y)R(x)E_n(x)=\id.
\ee

Thus, we have
\be
\ba{rl}
\lim_{n\ra\infty}\frac{1}{cn^{2/3}}\hat K_n(x,y)=&\frac{1}{2\pi i(u-v)}\times (-2\pi i)\left(\ba{cccc}-y'_0(v)&y_0(v)&0&0\ea\right)\left(\ba{c}y_0(u)\\y'_0(u)\\0\\0\ea\right)\\
{}=&\frac{y_0(u)y'_0(v)-y'_0(u)y_0(v)}{u-v}.
\ea
\ee
Since $y_0=Ai$, then we prove the first equation as $u,v<0$. The rest is similarly, we just use the another relation between $P$ and $\tilde P$.
\end{proof}


{\bf Acknowledgements} This work of Xu was supported by Shanghai Sailing Program
supported by Science and Technology Commission of Shanghai Municipality
under Grant NO. 15YF1408100 and the Hujiang Foundation of China (B14005). Most
of this work was done during the author(Xu) visited Indiana
University-Purdue University, Indianapolis. And the author also
would like to thank Prof. Bleher for his advise and help and
arranging this visit. Fan was support by grants from the National
Science Foundation of China (Project No. 11271079).
Yang thank the Macau Science and Technology Development Fund for generous support: MYRG2014-00011-FST  from Macau

\end{document}